%% file: main.tex
\documentclass[10pt, conference, letterpaper]{IEEEtran}

\input{input.tex}

\usepackage{mathtools,dsfont}
\usepackage{psfrag,bbm,graphicx}
\usepackage{algorithm,algorithmic}
\usepackage{comment,balance}
\usepackage{epstopdf}
\usepackage{epsfig}
\usepackage{multicol}
\usepackage{multirow, cite}
\usepackage{amsfonts}
\usepackage{color}

\usepackage{amsfonts,amsmath,amssymb,mathrsfs}
\usepackage{graphics, graphicx}

\usepackage{enumitem}
\usepackage[normalem]{ulem}
\usepackage{algorithm,algorithmic}

\usepackage{tikz}
\usetikzlibrary{shadows,patterns,shapes,backgrounds,automata}
\usetikzlibrary{fit,decorations.pathmorphing,plotmarks} 
\tikzstyle{every picture}+=[remember picture]

\newcommand{\ignore}[1]{}

\newcommand{\expect}[1]{{\mathbb E} \left[ #1 \right]}
\newcommand{\prob}[1]{{\mathbb P} \left( #1\right)}

\def\0v{\mathcal{U}_{0}}
\def\1v{\mathcal{U}_{1}}

\def \OO {\mathrm{O}}
\def \oo {\mathrm{o}}

\newtheorem{theorem}{Theorem}
\newtheorem{lemma}{Lemma}
\newtheorem{remark}{Remark}
\newtheorem{corollary}{Corollary}
\newtheorem{assumption}{Assumption}

\title{Paging with Multiple Caches}
\author{Rahul Vaze\\
	School of Technology and Computer Science \\
	Tata Institute of Fundamental Research \\
	       Email: vaze@tcs.tifr.res.in
	\and
	Sharayu Moharir\\
	Department of Electrical Engineering \\
	Indian Institute of Technology, Bombay \\
	Email: sharayum@ee.iitb.ac.in
}

\begin{document}
\maketitle

\begin{abstract}
	Modern content delivery networks consist of one or more ``back-end" servers which store the entire content catalog, assisted by multiple ``front-end" servers with limited storage and service capacities located near the end-users. Appropriate replication of content on the front-end servers is key to maximize the fraction of requests served by the front-end servers.
		
	Motivated by this, a multiple cache variant of the classical single cache paging problem is studied, which is referred to as the Multiple Cache Paging (MCP) problem. In each time-slot, a batch of content requests arrive that have to be served by a bank of caches, and each cache can serve exactly one request. If a content is not found in the bank, it is fetched from the back-end server, and one currently stored content is ejected, and counted as  `fault'. 
	
As in the classical paging problem, the goal is to minimize the total number of faults. The competitive ratio of any online algorithm for the MCP problem is shown to be unbounded for arbitrary input, thus concluding that the MCP problem is fundamentally different from the classical paging problem. Consequently, stochastic arrivals setting is considered, where requests arrive according to a known/unknown stochastic process. It is shown that near optimal performance can be achieved with simple policies that  require no co-ordination across the caches. 
\end{abstract}
\section{Introduction}
To serve the ever increasing video traffic demand over the internet, many Video on Demand (VoD) services like Netflix \cite{Netflix} and Youtube \cite{Youtube} use a two-layered content delivery network \cite{Netflix_openconnect}. The network consists of a back-end server which stores the entire catalog of contents offered by the service and multiple front-end servers, each with limited service and storage capacity, located at the `edge' of the network, i.e., close to the end users. Content can be fetched from the back-end server and replicated on the front-end server to serve user requests. Each such fetch adds to the cost the network pays to serve user requests. Compared to the back-end server, the front-end servers can serve user requests more efficiently due to their proximity to the users. The motivation behind this hierarchical architecture is to serve most of the requests using the front-end servers, thus, reducing the load on the back-end server and therefore the network backbone. The task of allocating incoming requests to front-end servers and the content replication strategy, i.e., which content is stored on the front-end servers form an important part of the system architecture and have been studied in various settings \cite{SGSS14, LLM12,XT13,LLM13,aaglr10, maddah2014fundamental, maddah2013decentralized, niesen2014coded}.

Motivated by this, we study a multiple cache variant of the single cache paging problem \cite{BookRaghavan}, which we call the Multiple Cache Paging (MCP) problem. The MCP problem is defined as follows: we consider a time-slotted system where in each time slot, similar to \cite{maddah2014fundamental, maddah2013decentralized, niesen2014coded}, a batch of requests arrives such that the number of requests in each batch is equal to the number of caches/servers. We assume that each cache can serve at most one request at a time. The task is to match these requests to the caches. If a request is matched to a server which does not have the requested content stored on it, the requested content is fetched from the back-end server, and replicated on the cache by ejecting one of the currently stored contents. Every such ejection/fetch is referred to as a fault. The MCP problem has two inter-related challenges: matching requests and ejecting content, in an online manner, i.e., without knowing the future arrival sequence. Our objective is to design online algorithms which minimize the number of faults made to serve incoming requests. The classical paging problem is a special case of the MCP problem with only one single cache \cite{sleator1985amortized}.  

We first study the most general case where the request arrival sequence can be arbitrary. To characterize the performance of an online algorithm, we consider the metric of competitive ratio that is defined as the ratio of the faults made by the online policy and the offline optimal policy, maximized over all input sequences. The competitive ratio is a worst case guarantee on the performance of an online algorithm and therefore can be too pessimistic, however, for the classical paging problem, it is known that there are deterministic \cite{sleator1985amortized} as well as randomized algorithms \cite{fiat1991competitive} with bounded competitive ratios. 

We also study the stochastic setting, where requests are generated by a known/unknown stochastic process. We focus on two specific cases: \emph{(i)}  iid Zipf distribution, i.e., content popularity is heavy-tailed where the popularity of the $i^{\text{th}}$ most popular content is proportional to $i^{-\beta}$ for a constant $\beta>0$. The Zipf distribution, is known to be a good match for content popularity in VoD services \cite{Gill07,BC99, YZ06, IRF04, VA02}, \emph{(ii)} a specific correlated content access model, where contents requested by each of the stream are allowed to be correlated across time.

\subsection{Contributions}

The main contributions of our work are as follows:

\begin{enumerate}
	
	\item \textbf{Adversarial setting:}  We show that  for the MCP problem, the competitive ratio of any online algorithm is unbounded (Theorem \ref{theorem:competitive_ratio}). This is surprising, since for the classical paging problem with single cache \cite{BookRaghavan}, there are deterministic algorithms with competitive ratios at most equal to the size of the cache memory \cite{sleator1985amortized}. We obtain this result by exploiting the additional matching decision required in solving MCP compared to the single cache problem, and thus conclude that the MCP problem is fundamentally different from the classical paging problem. 
	
	\item \textbf{Stochastic setting with known popularity:} We consider a policy called Cache Most Popular (CMP), that has been proposed before in \cite{LLM12}, however, in a static scenario. The key idea of the CMP policy is that each cache stores the same most popular contents. We show that even with a fixed matching between requests and caches, and even if each cache makes its ejection decisions independently, CMP is optimal for Zipf distribution with $\beta <1$, while for $\beta> 1$ its competitive ratio grows as $m^{\beta-1}$, where $m$ is the number of servers, however, is independent of the size of each cache $k$ (Theorem \ref{thm:zipf_iid}).  Note that since $\beta <2$ (typically $1.2$) for most applications, our results show that CMP policy's competitive ratio grows sub-linearly in the number of servers. Similar result is also true for CMP when the stochastic input is correlated, for a specific correlation model described later in detail.
	
	\item \textbf{Stochastic setting with unknown popularity distribution:} For this setting, we consider the widely popular Least Recently Used (LRU) policy \cite{sleator1985amortized}, and show that even with a fixed matching between requests and caches, and no coordination across caches, the competitive ratio of the LRU policy is $m^{\beta-1} (\ln k)^{\left(2-\frac{2}{\beta}\right)}$, when the underlying distribution is Zipf with parameter $\beta > 1$ and $k$ is the size of each cache. Identical result is also obtained for the correlated input model. Thus, comparing the known and unknown popularity distribution, the price to pay for lack of knowledge of distribution is $(\ln k)^{\left(2-\frac{2}{\beta}\right)}$. For technical difficulties, we do not get a result for $\beta < 1$ with unknown distribution.
\end{enumerate} 

In the context of the motivating application of content delivery, our results show 
that close to optimal performance can be achieved even when the routing (matching requests to caches) is done in a decentralized manner (actually fixed ahead of time, which could model any geographical or load balancing constraint), 
and there is no co-ordination between the front-end servers. 

\subsection{Related Work}
\label{subsec:related_work}

\subsubsection{Paging Problem} The paging problem with a single cache (size $k$) has been widely studied in the online algorithms literature \cite{BookRaghavan}. 
For this problem, ejecting the least recently used page (LRU), and first-in-first-out (FIFO) are known to be optimal deterministic online algorithms \cite{sleator1985amortized} with a competitive ratio of $k$. Moreover, the optimal competitive ratio of 
$\ln k$ can be achieved by a  randomized version of LRU \cite{fiat1991competitive}. The paging problem with a single cache has also been studied for Markovian arrivals in \cite{karlin1992markov} and close to optimal algorithms have been proposed. 

\subsubsection{Content Replication Policies} The problem of content replication and request matching for content delivery networks have been studied in \cite{SGSS14, LLM12,XT13,LLM13,aaglr10, maddah2014fundamental, maddah2013decentralized, niesen2014coded}. The setting where content popularity is known is studied in \cite{ LLM12,XT13,LLM13}, while \cite{SGSS14,aaglr10} focus on the setting where content popularity is unknown. Static content replication policies have been studied in \cite{LLM12,XT13,Whitt07,aaglr10}, where the decision about which contents to store in which cache are made ahead of time, and no ejection or
cache updation is allowed once the requests start to arrive. Optimal adaptive content replication polices for the setting where the number of front-end servers is large have been proposed in \cite{SGSS14,LLM13}. The importance of coded caching for optimal performance has been shown in \cite{maddah2014fundamental, maddah2013decentralized, niesen2014coded}.

In a major departure from prior work \cite{SGSS14, LLM12,XT13,LLM13} that study the asymptotic setting where the number of front-end servers scale linearly with the number of contents, we look at the setting where the number of front-end servers scale slower than the number of contents, for example, the number of front-end servers can be a constant. 
In terms of the proposed algorithms, a key difference between our work and all the prior work is that unlike CMP and LRU, all the algorithms proposed in \cite{SGSS14, LLM12,XT13,LLM13,Whitt07,aaglr10} are centralized where both routing and content replication decisions are made in a centralized manner. Thus, our algorithms are easily scalable and can satisfy any matching restrictions because of geographical or load balancing constraints.

\subsection{Modeling Details}
We consider the batch processing model for CDNs, similar to \cite{maddah2014fundamental, maddah2013decentralized, niesen2014coded}, where a 'batch' of requests arrives at each time instant such that the number of requests in each batch is equal to the number of caches/servers. To model a delay-sensitive and QoS guarantee setting, where  jobs/requests are neither queued nor dropped, each request is required to be served independent of its popularity by one of the servers.


\section{System Model}
Consider a bank of $m$ caches, each of which can store $k$ out of the $n$ contents in the catalog.  
The arrivals are time-slotted, where in each slot, $m$ parallel requests arrive to the system that have to be allocated/matched to/with the $m$ caches, such that each cache is allocated exactly one request. Once the requests are matched to the caches, if the requested content is not stored in the corresponding cache, the cache replaces one of the currently stored contents with the requested content to serve the request. We refer to such a replacement as a {\it fault}. We call this problem as multiple cache paging (MCP) and the goal is to design an online algorithm (without knowing the future requests) which minimizes the number of faults made to serve the requests, where both the  matching and the ejection problems have to be jointly solved. 
For $m=1$, this is the classical paging problem for which we summarize the results in Theorem \ref{thm:singlecache} and \ref{theorem:competitive_ratio}, where no matching decisions are needed.

To model the most general setting, we first consider the case that the content request sequence is arbitrary, i.e., it can even be chosen by an adversary. 
A natural measure for performance analysis of online algorithms is the competitive ratio, where the competitive ratio for an online algorithm ALG  $\rho_{\text{ALG}}$ is defined as
\begin{eqnarray*}
	\rho_{\text{ALG}} = \max_{R} \dfrac{\sfF_{\text{ALG}}(R)}{\sfF_{\text{OPT}}(R)},
\end{eqnarray*}
where, $R$ is the content request arrival sequence, and $\sfF_{\text{ALG}}(R)$ and $\sfF_{\text{OPT}}(R)$ are the number of faults made by algorithm ALG and the optimal offline algorithm, respectively. The offline optimal algorithm knows the entire sequence $R$ a-priori. We use the following notation throughout the paper. The request sequence $R = ({\br}(t))$, where each $\br$ is a $m$-length content request vector $\br = (r_1, \dots, r_m)$. The contents of cache $j$ at time $t$ are denoted as $\sfc_j(t)$, and $\sfC(t) = (\sfc_j(t), j=1,\dots m)$.

Next, we consider the stochastic setting, where the arrival sequence is generated by an underlying distribution. In this setting, the goal is to minimize the expected number of faults in $T$ time-slots, denoted by $\expect{\sfF_{\text{ALG}}(T)}$.
\section{Preliminaries}

Under the arbitrary input setting, the LRU algorithm has a competitive ratio of $k$ when $m=1$, which is also a 
lower bound on the competitive ratio. For completeness sake, and to keep the main negative result (Theorem \ref{theorem:competitive_ratio}) of this paper in perspective, we state the optimal optimal competitive ratio result for the single cache problem as follows. 

\begin{theorem}\label{thm:singlecache} LRU is an optimal deterministic online algorithm for MCP problem with a single cache, and its competitive ratio is $k$ \cite{BookRaghavan}.
\end{theorem}

Theorem \ref{thm:singlecache} tells us that the worst case input setting for a single cache is not degenerate, and the performance of LRU is bounded even when an adversary can choose the request sequence sequence. We will show in Theorem \ref{theorem:competitive_ratio} that this is not the case for the MCP problem with $m>1$, and show that the lower bound on the competitive ratio of any algorithm is unbounded.

\section{Lower Bound on the Competitive Ratio for MCP problem}

In this section, we present our first main result that shows that the competitive ratio of any online algorithm for solving the MCP is unbounded, which is in contrast to MCP with a single server.

\begin{theorem} 
	\label{theorem:competitive_ratio}
	The competitive ratio of any online algorithm for solving MCP with $m>1$ is unbounded.
\end{theorem}
\begin{proof}
	We will prove it for the specific case of $m=2$ and $k=4$, which suffices for proving the Theorem. We will produce a sequence of input requests for which no matter what  an online algorithm does, the number of faults is unbounded, while an optimal offline algorithm makes only two faults.
	First we state two rules that any online algorithm must follow, since otherwise there is an easy construction of 'bad' sequences for which the competitive ratio is unbounded. For lack of space, the construction is omitted, however, it is quite natural and easy to do so.  
	\begin{enumerate}
		\item  Let the weight of an edge between the $i^{th}$ component of the input vector $\br(t)$  and cache $j$ be defined to be $1$ if cache $j$ contains input $i$, otherwise zero. Let $M(t)$ be a set of maximum weight matchings of the input vector $\br(t)$ and the contents of cache $\sfC(t)$ at time $t$.  
		Then any online algorithm at time $t$ must serve the input vector $\br(t)$ via any one of the maximum weight matchings.
		\item Once an online algorithm decides on a matching (which content will be served by which cache), then it should always eject the oldest arrived content in the cache. It will be clear from the following input construction, that  if the online algorithm ejects the most recently requested content to serve the current request, then the adversary can repeatedly ask for a two tuple of contents $a,b$, for which any online algorithm will repeatedly fault on each request. Similar argument can be extended for any other policy, which does not eject the oldest arrived content. More details can be found in \cite{BookRaghavan}. 
	\end{enumerate}
	
	Consider $m=2$ and $k=4$. Assume that the contents of the two caches at certain time are 
	$\mathsf{c}_1 = \{x_1, a_2, a_3, a_4\}$, and $\mathsf{c}_2 = \{x_2, b_2, b_3, b_4\}$. Without loss of generality, since files are only place holders, we will follow a convention that $a_k, b_k$ came earlier than $a_{k-1},b_{k-1}$ in cache $1$ and $2$, respectively, for $k=2,3,4$. Moreover, we assume that $x_i, i=1,2$ are the oldest contents in cache $i$, respectively, at this time.
	
	Let at next time slot, the request vector be $\{a_1, b_1\}$. Not knowing the future input sequence, let any online algorithm $A$ make the following allocation, serve $a_1$ from cache $1$, and $b_1$ from cache $2$, by evicting $x_1$ and $x_2$, respectively. Then $\mathsf{c}_1(A) = \{a_1, a_2, a_3, a_4\}$, and $\mathsf{c}_2(A) = \{b_1, b_2, b_3, b_4\}$. 
	Moreover, let the optimal offline algorithm, knowing the future input, serve $a_1$ from $\mathsf{c}_2$, and $b_1$ from $\mathsf{c}_1$, by evicting $x_2$ and $x_1$, respectively. Thus, the updated contents of an optimal offline algorithm are $\mathsf{c}_1(\mathsf{opt}) = \{b_1, a_2, a_3, a_4\}$, and $\mathsf{c}_2(\mathsf{opt}) = \{a_1, b_2, b_3, b_4\}$. Note that for competitive ratio definition, the adversary is allowed to give inputs depending on the current state of the algorithm $A$. If suppose, $A$ makes the opposite allocation, then we can replace the role of $a_1$ and $b_1$ in future input sequence, and get the same result as $a_1$ and $b_1$ are only place holders, as will be evident as follows.
	\begin{table}[h]
		\centering
		\begin{tabular}{| c | c | c | c | c | c | c | c | c |}
			\hline
			\text{input} & \multicolumn{4}{|c|} {$\mathsf{c}_1$} & \multicolumn{4}{|c|}{$\mathsf{c}_2$}  \\ \hline
			& $a_1$ & $a_2$ & $a_3$ & $a_4$  & $ b_1$ & $b_2$ & $b_3$ & $b_4$ \\\hline
			($a_1$,$a_2$)&  &  &  &  &  &    &  &   $a_2$  \\\hline
			($b_1$, $b_2$)  &  &  &  &$b_2$  & &    &  &   \\\hline
			($a_1$, $a_3$)&  &  &  &  & &    &   $a_3$ &  \\\hline
			($a_3$, $b_3$)&  &  &  &  & &     $b_3$ & & \\\hline
			($a_1$, $a_4$)  &  &  &  &  &     $a_4$&  & &  \\\hline
			($a_3$, $b_4$)&  &  &  &  & &    &  & $b_4$ \\\hline
			&  &  &  &  & &    &  &  \\\hline
			($a_1$, $a_2$)&  &  &  &  &     &  & $a_2$&  \\\hline
			($b_1$, $b_2$)&  &  &  &  &    &   $b_1$& & \\\hline
			($a_1$, $a_3$)&  &  &  &   &   $a_3$&  & & \\\hline
			($a_3$, $b_3$)&  &  &  & &    &  & &$b_3$  \\\hline
			($a_1$, $a_4$)&  &  &    & &    &  &$a_4$ & \\\hline
			($a_3$, $b_4$)&  &  &    & &    & $b_4$ & &  \\\hline
			&  &  &  &  & &     & &  \\\hline
			($a_1$, $a_2$)&  &  &  &   &    $a_2$&  & &  \\\hline
			($b_1, b_2$) &  &  &  &  & &    & &$b_1$ \\\hline
			($a_1, a_3$) &  &  &  &  & &      &$a_3$ &  \\\hline
			($a_3, b_3$)&  &  &  &  & &     $b_3$ & &  \\\hline
			($a_1, a_4$)&  &  &  &  &  $a_4 $  &  & & \\\hline
			($a_3, b_4$) &  &  &  &  &     &  & & $b_4$ \\\hline
			\hline
			\hline
		\end{tabular}
		\\ 
		\vspace{0.051in}
		Table $1$.
	\end{table}

	Now we construct a sequence of inputs so that the optimal offline algorithm does not have to incur any fault, while  any online algorithm $A$ makes at least one fault at each subsequent request. 
	Let  $\sigma=\{ (a_1, a_2), (b_1, b_2), (a_1, a_3), (b_2, b_3), (a_1, a_4), (b_4, b_2)\}$. 
	From hereon, we will use $R = \{\sigma, \sigma, \dots\}$ that uses $\sigma$  repeated infinitely many times, as the request vector sequence. 
	

	
	Without loss of generality, we follow that if the input is $(*_i, *_j), i<j, * \in \{a,b\}$ and $*_i,*_j \in \mathsf{c}_u(A)$ and $*_i,*_j \notin \mathsf{c}_v(A)$, we follow that $*_i$ is served from cache $u$, and $*_j$ is served from cache $v$ after evicting some content from $\mathsf{c}_v(A)$, since the content names are only place holders. 
	
	For input $R$ as described above, Table $1$ illustrates the evolving contents of the two caches with any online algorithm following rule 1 and 2 and our convention, where in any row we write the content that is entering the cache which replaces the oldest content of that cache placed in the corresponding column. For example, at time $2$, $b_2$ replaces $a_4$ in cache $1$.
	
	The main idea to notice from Table $1$ is that after one full run of $\sigma$, 
	contents of cache $1$, $\mathsf{c}_1$, are never changed and it always contains $\{a_1, a_2, a_3,b_2\}$. 
	The request sequence at each subsequent time is such that only one content can be served from $\mathsf{c}_1$ without any fault, while the other content is missing from $\mathsf{c}_2$ and at least one fault has to be made.  
	Since $\mathsf{c}_1$ is never updated, there is no way for any online algorithm to make cache contents $\mathsf{c}_1, \mathsf{c}_2$ equal to the offline algorithm's $\mathsf{c}_1(\mathsf{opt})$ and $\mathsf{c}_2(\mathsf{opt})$ upto a permutation. Moreover, it is easy to check that the full request sequence $R$ can be served by the optimal offline algorithm that has contents $\mathsf{c}_1(\mathsf{opt}) = \{b_1, a_2, a_3, a_4\}$, and $\mathsf{c}_2(\mathsf{opt}) = \{a_1, b_2, b_3, b_4\}$ without incurring any fault.

\end{proof}

{\it Discussion:} Theorem \ref{theorem:competitive_ratio} is a surprising result in light of Theorem \ref{thm:singlecache}, since it shows that even if there are only $2$ caches, no online algorithm can have a bounded competitive ratio. One would have hoped that the with $m>1$, the competitive ratio might grow as $mk$ following Theorem \ref{thm:singlecache}, but remarkably we show that MCP with $m>1$ is fundamentally different than the single cache problem, and the added decision of matching together with ejection makes it very hard for any online algorithm to stay close to the optimal offline algorithm. The main idea in proving Theorem \ref{theorem:competitive_ratio} is that if any online algorithm makes one mistake in matching the requests, then that algorithm can be forced to make repeated mistakes. 
This result also shows that MCP with $m>1$ and $k$-sized caches is not related to MCP with single cache and memory size of $mk$, since otherwise the LRU algorithm would achieve a competitive ratio of at most $mk$. 
Therefore for MCP with multiple caches, the arbitrary or adversarial request setting is too pessimistic/degenerate. Consequently, in the next section, we study the MCP under a stochastic arrival model, that is also well motivated from practical applications, to better understand its fundamental performance limits.

\section{I.I.D. Stochastic Arrivals}


\subsection{Known Popularity}
\begin{assumption}
	\label{ass:IID_arrivalprocess}
	I.I.D. Arrival Process \\
	The content catalog consists of $n$ contents denoted by $C_i$, $i=$ 1, 2, .., $n$. As before, vector $\br(t)$ ( $r_1$, $r_2$, .., $r_m$) arrives at the beginning of each time-slot, where each individual  request is i.i.d., with 
	$$ \prob{r_j= C_i} = p_i, \forall i, j.$$
	Without loss of generality, we assume that $p_1 \geq p_2 \geq ... \geq p_n$.
	
	Content popularity follows the Zipf Law \cite{zipf-wiki} for many  VoD services \cite{Gill07,BC99, YZ06, IRF04, VA02}. Zipf's law states that if contents are indexed in decreasing order of popularity, i.e., $C_i$ is the $i^{\text{th}}$ most popular content, the popularity of $C_i$: $p_i \propto i^{-\beta}, \text{ for a constant }\beta > 0.$ 
\end{assumption}

We consider a distributed caching policy called Cache Most Popular (CMP). Under this policy, the matching between requests and cache is fixed, i.e., in each time-slot, request $j$ ($r_j$) is served by cache $j$. Each cache stores the same $k$ most popular contents to begin with. If $r_j$ is not stored in the cache, Cache $j$ replaces the least popular currently cached content with $r_j$ to serve the request. Refer to Figure \ref{fig:MYOPIC} for a formal definition of the CMP policy. Let $C_i$ be arranged in decreasing order of popularity.

\begin{figure}[h]
	\hrule
	\vspace{0.1in}
	\begin{algorithmic}[1]
		
		\STATE Initialize: Store contents $C_1$, $C_2$, .., $C_k$ in each cache. 
		\STATE On arrival, $\forall j$, allocate $r_j$ to Cache $j$ \textbf{do},
		\IF {$r_j$ not stored in cache $j$,}
		\STATE replace the least popular content in cache $j$ with $r_j$.
		\ENDIF
	\end{algorithmic}
	\vspace{0.1in}
	\hrule
	\caption{CACHE MOST POPULAR (CMP) -- \sl An adaptive caching policy which caches the most popular content.}
	\label{fig:MYOPIC}
\end{figure}

\begin{remark} 
	\label{remark:CMP_properties}
	The CMP policy has the following properties:
	\begin{enumerate}
		\item[(i)] The routing decision (matching requests and caches) is fixed, and therefore independent of the current state of the caches.
		\item[(ii)] Each cache makes its decisions (which content to eject) independent of the other caches in the system.  
		\item[(iii)] At any instant, the $k-1$ most popular contents are stored by all $m$ caches. 
		\item[(iv)] At any instant, the maximum number of unique contents stored across all $m$ caches is $k+m-1$.		
	\end{enumerate}
\end{remark}

We now compare the performance of the CMP policy with the offline optimal policy (OPT).

\begin{theorem}
	\label{thm:zipf_iid}
	The competitive ratio of CMP with Zipf distribution is as follows.
	\noindent If $k>1$ and $mk \leq \alpha n$ for a constant $\alpha < 1$, 
	\begin{enumerate}
		\item[(i)] If $\beta \leq 1$, $\displaystyle \limsup_{n \rightarrow \infty} \dfrac{\expect{\sfF_{\text{CMP}}(T)}}{\expect{\sfF_{\text{OPT}}(T)}} = \Theta(1)$.
		\item[(ii)] If $\beta > 1$, 
		$\displaystyle \limsup_{n \rightarrow \infty} \dfrac{\expect{\sfF_{\text{CMP}}(T)}}{\expect{\sfF_{\text{OPT}}(T)}} = \OO(m^{\beta-1}).$
	\end{enumerate}
	If $k>1$ and $\beta < 1$, if $mk = \oo(n)$,  $\displaystyle  \limsup_{n \rightarrow \infty} \dfrac{\expect{\sfF_{\text{CMP}}(T)}}{\expect{\sfF_{\text{OPT}}(T)}} = 1$.
	
\end{theorem}

We thus conclude that for $\beta<1$, CMP is asymptotically optimal if $mk = \oo(n)$, i.e., 
the total cache storage is a vanishing fraction of the total content catalog, which is a realistic setting since typically the number of contents is much larger than the total storage space. This includes the case when $\beta = 0.8$ which is a good match for web pages \cite{fricker2012impact}. Moreover, notice that we get this result even though CMP is a distributed caching policy which uses a fixed matching between requests and caches and each cache makes its decisions independently. 

For $\beta>1$, the competitive ratio of CMP scales sub-linearly with the number of servers $m^{\beta -1}$, but importantly is independent of the size of each cache $k$. For $\beta = 1.2$, which has been observed to be a good match for Video on Demand (VoD) services \cite{fricker2012impact}, the competitive ratio scales as $m^{0.2}$, which is still reasonable.

\subsection{Unknown content popularity distribution}
To generalize the model even more, we assume that the underlying content popularity is Zipf distributed but that is unknown to the algorithm. In this case, we consider that each cache implements the Least Recently Used (LRU) policy. As in the CMP policy, the matching between requests and cache is fixed, i.e., in each time-slot, request $j$ ($r_j$) is served by cache $j$. If $r_j$ is not stored in the cache, Cache $j$ replaces the least recently used cached content with $r_j$ to serve the request. Refer to Figure \ref{fig:LRU} for a formal definition of the LRU policy. 

\begin{figure}[h]
	\hrule
	\vspace{0.1in}
	\begin{algorithmic}[1]
		\STATE Initialize: Start with $m$ empty caches.  
		\STATE On arrival, allocate $r_j$ to Cache $j$ \textbf{do},
		\IF {$r_j$ not stored in Cache $j$,}
		\IF {Cache $j$ already has $k$ contents}
		\STATE replace the least recently requested content with $r_j$.
		\ELSE
		\STATE store $r_j$ in Cache $j$
		\ENDIF
		\ENDIF
	\end{algorithmic}
	\vspace{0.1in}
	\hrule
	\caption{LEAST RECENTLY USED (LRU) -- \sl An adaptive caching policy which changes the content cached in a greedy manner without using the knowledge of request arrival statistics.}
	\label{fig:LRU}
\end{figure}

\begin{theorem}
	\label{theorem:unknown_zipf}
	When the underlying distribution is Zipf but that is unknown to the algorithm, the LRU based policy's competitive ratio for large enough $k$ is given by, 
	\begin{eqnarray*}
		\displaystyle \displaystyle \limsup_{n \rightarrow \infty} \dfrac{\expect{\sfF_{\text{LRU}}(T)}}{\expect{\sfF_{\text{OPT}}(T)}} = \OO(m^{\beta-1}(\log k)^{2-2/\beta} ).
	\end{eqnarray*}
\end{theorem}
%
For technical reasons, we do not have corresponding result for $\beta < 1$.

\section{Time-Correlated Arrivals} Finally, we consider the case when content requests are time-correlated, and propose a 
model that captures the successive video file demands in practical systems. The model is not the most general, 
however, is a tradeoff between a most general model and an analytically tractable one. 
\vspace{0.1in}

Model: Recall that we have $m$ streams of requests. We assume that each stream is generated as follows: 

\begin{assumption} Time-Correlated Arrivals \\
	\label{ass:correlated_general}
	Let $S$ be the set of all possible strings of contents of arbitrary length in the catalog, and $X(s)$, $\forall s \in S$ be a distribution on the elements of $S$. 
	\begin{enumerate}
		\item[--] At time $0$, a string of length $b_0$ is picked according to $X$. These $b_0$ contents are requested in sequence, 1 content per time-slot in the next $b_t$ time-slots. 
		\item[--] At the end of the previous request sequence i.e. at time $\sum_{i=0}^t b_i$, a new sequence of length $b_{t+1}$  is picked from the same distribution. The process is repeated at end of each such sequence.
		\item[--] The sequence is revealed to the system in an online manner (1 content per time-slot).
	\end{enumerate} 
	We call each string of length $b_i$ as a sub-sequence, which together make the whole sequence $r_j$ for each of the $m$ streams.
\end{assumption}

\noindent For each content $i$, we define a new quantity $\tilde{p}_i$ as follows:
\begin{eqnarray*}
	\tilde{p}_i = \expect{\text{\# of times content }C_i \text{ appears in the sub-sequence}}.
\end{eqnarray*}
Without loss of generality, we assume that $\tilde{p}_1 \geq \tilde{p}_2 \geq ... \geq \tilde{p}_n$. Let $Z(T)$ be the number of request sub-sequences which end during or before time-slot $T$ summed across the $m$ streams, i.e., for each 
stream $j$, define  $t_j = \max\{t: \sum_{i=0}^t b_i(j) < T\}$. Then $Z(T) = \sum_{j=1}^m t_j$. Moreover, 
define $L$ to be the length of a requested sub-sequence for any of the $m$ streams. 

Then we have the following result on the performance of CMP for this setting, where the system is in operation for total of $T$ slots.
\begin{theorem}
	\label{theorem:markov_competitiveratio}
	Under Assumption \ref{ass:correlated_general}, the competitive ratio of the CMP policy until time $T$ is given by 
	\begin{eqnarray*}
		\dfrac{\expect{\sfF_{\text{CMP}}(T)}}{\expect{\sfF_{\text{OPT}}(T)}} \leq \bigg( 1 + \dfrac{m}{\expect{Z(T)}} \bigg) \dfrac{\sum_{i=k}^{n} \tilde{p}_i }{\sum_{i=m(k+\expect{L})+1}^{n} \tilde{p}_i}.
	\end{eqnarray*}
\end{theorem}
Theorem \ref{theorem:markov_competitiveratio} is a general result for any correlated distribution satisfying Assumption \ref{ass:correlated_general}. 
To get more precise results, we next consider a specific correlated distribution, which is a partitioned Zipf correlated distribution and evaluate the performance of CMP with respect to it.  
\begin{assumption} The arrival process satisfies the following properties:
	\label{ass:recommendation_system}
	\begin{enumerate}
		\item[--] The $n$ contents are divided into $n/b$ groups of $b$ contents each, where $b$ is a constant $>1$. 
		\item[--] Each request sub-sequence consists of requests for contents in only one group and
		$\prob{\text{requested contents belong to group } l} \propto l^{-\beta}$,
		for a constant $\beta > 0$.
		\item[--] The request sub-sequence consists of $\min\{y, b\}$ distinct contents in the group, chosen uniformly at random, where $y$ is a Geometric random variable with parameter $\gamma$. 
	\end{enumerate}
\end{assumption}
With this model, sub-sequences within a group are correlated, while across groups, they follow a Zipf distribution. This models clustered content popularity, where certain clusters are more popular than others, while within each cluster there is correlation. With this model, we can get the following result for the CMP, with time horizon of $T$ slots.

\begin{corollary}
	\label{corollary:markov}
	If $k>1$, $mk = \oo(n)$, $b = \oo(k)$, and under Assumption \ref{ass:recommendation_system}, for $T \geq b$,
	\begin{enumerate}
		\item[(i)] If $\beta < 1$, $\displaystyle  \limsup_{n \rightarrow \infty} \dfrac{\expect{\sfF_{\text{CMP}}(T)}}{\expect{\sfF_{\text{OPT}}(T)}} = \bigg( 1 +  \bigg \lfloor \frac{T}{b} \bigg\rfloor ^{-1} \bigg)$.
		\item[(ii)] If $\beta = 1$, 
		$\displaystyle \limsup_{n \rightarrow \infty} \dfrac{\expect{\sfF_{\text{CMP}}(T)}}{\expect{\sfF_{\text{OPT}}(T)}} = \bigg( 1 + \bigg \lfloor \frac{T}{b} \bigg\rfloor ^{-1} \bigg) \Theta(1)$.
		\item[(iii)] If $\beta > 1, $\\
		$\displaystyle \limsup_{n \rightarrow \infty} \dfrac{\expect{\sfF_{\text{CMP}}(T)}}{\expect{\sfF_{\text{OPT}}(T)}} = \bigg( 1 + \bigg \lfloor \frac{T}{b} \bigg\rfloor ^{-1} \bigg) \OO(m^{\beta-1})$.
	\end{enumerate}
\end{corollary}

Comparing the performance of CMP with correlated Zipf distribution with respect to the Zipf (Theorem \ref{thm:zipf_iid}), we see that there is only an additional penalty term of $\big( 1 + \big \lfloor \frac{T}{b} \big\rfloor ^{-1} \big)$. Thus, the penalty decreases with time horizon $T$, but increases with the size of $b$, which is expected, since larger $b$ means more correlation.

Finally, we characterize the performance of LRU when the underlying distribution is  the correlated Zipf distribution, which is unknown to the algorithm.

\begin{theorem}
	\label{theorem:markov_unkown}
	If $mk = \oo(n)$, $\beta>1$, under Assumption  \ref{ass:recommendation_system} for $k$ large enough, 
	\begin{eqnarray*}
		\displaystyle \limsup_{n \rightarrow \infty} \dfrac{\expect{\sfF_{\text{LRU}}(T)}}{\expect{\sfF_{\text{OPT}}(T)}} = \bigg( 1 + \bigg \lfloor \frac{T}{b} \bigg\rfloor ^{-1} \bigg) \OO(m^{\beta-1}(\log k)^{2-2/\beta} ).
	\end{eqnarray*}
\end{theorem}
For lack of space, we omit the proof which follows similarly to the Proof of Theorem \ref{theorem:unknown_zipf}. 

\section{Simulations}

In Figure \ref{fig:diff_n_zipf}, we plot the upper bound on the competitive ratio of the CMP policy for the Zipf distribution as a function of the number of contents $(n)$ and the size of each cache $(k)$ with $m=10$. For all plots, we use $10^4$ time-slots. To evaluate the competitive ratio performance of any policy, we use the lower bound derived in Lemma \ref{lemma:lowerbound_OPT}.
As expected from Theorem \ref{thm:zipf_iid}, we see that the performance of CMP improves as $n$ increases as compared to $mk$, and worsens as the Zipf parameter $(\beta)$ increases.
\begin{figure}[h]
	\begin{center}
		\includegraphics[scale=0.24]{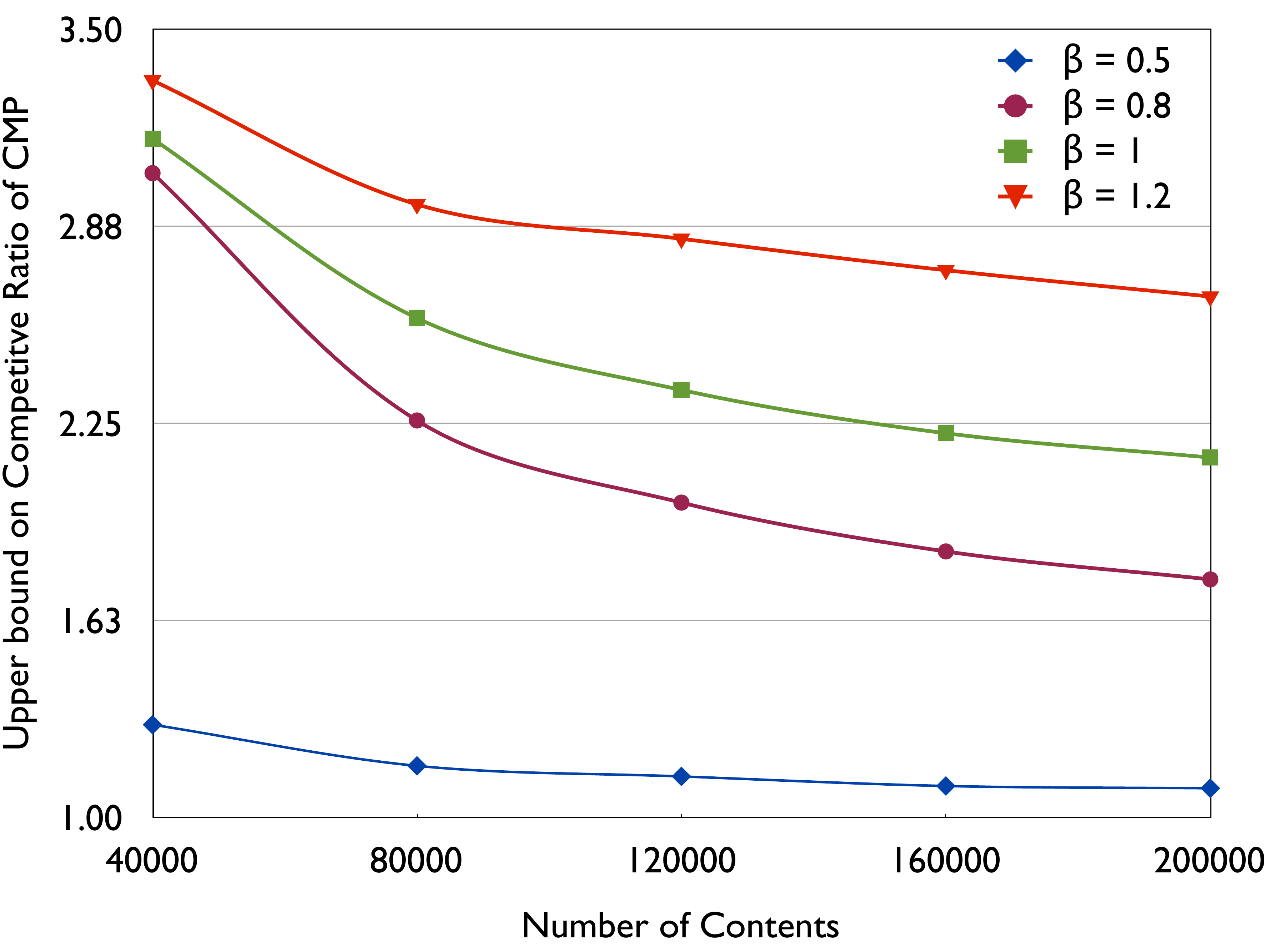}
		\caption{\sl The upper bound on the competitive ratio of the CMP policy for different number of contents $(n)$ and different values of the Zipf parameter $(\beta)$. \label{fig:diff_n_zipf}} 
	\end{center}
\end{figure}

In Figure \ref{fig:diff_m_zipf}, we plot the upper bound on the competitive ratio of the CMP policy for different number of caches $(m)$ and Zipf parameter $(\beta)$, where we fix the number of contents $(n)$ to $10000$. As expected, for fixed values of $n$ and $\beta$, the performance of CMP gets worse as the number of caches $(m)$ increases, and for fixed values of $n$ and $m$, smaller values of $\beta$ lead to better performance.  

\begin{figure}[h]
	\begin{center}
		\includegraphics[scale=0.24]{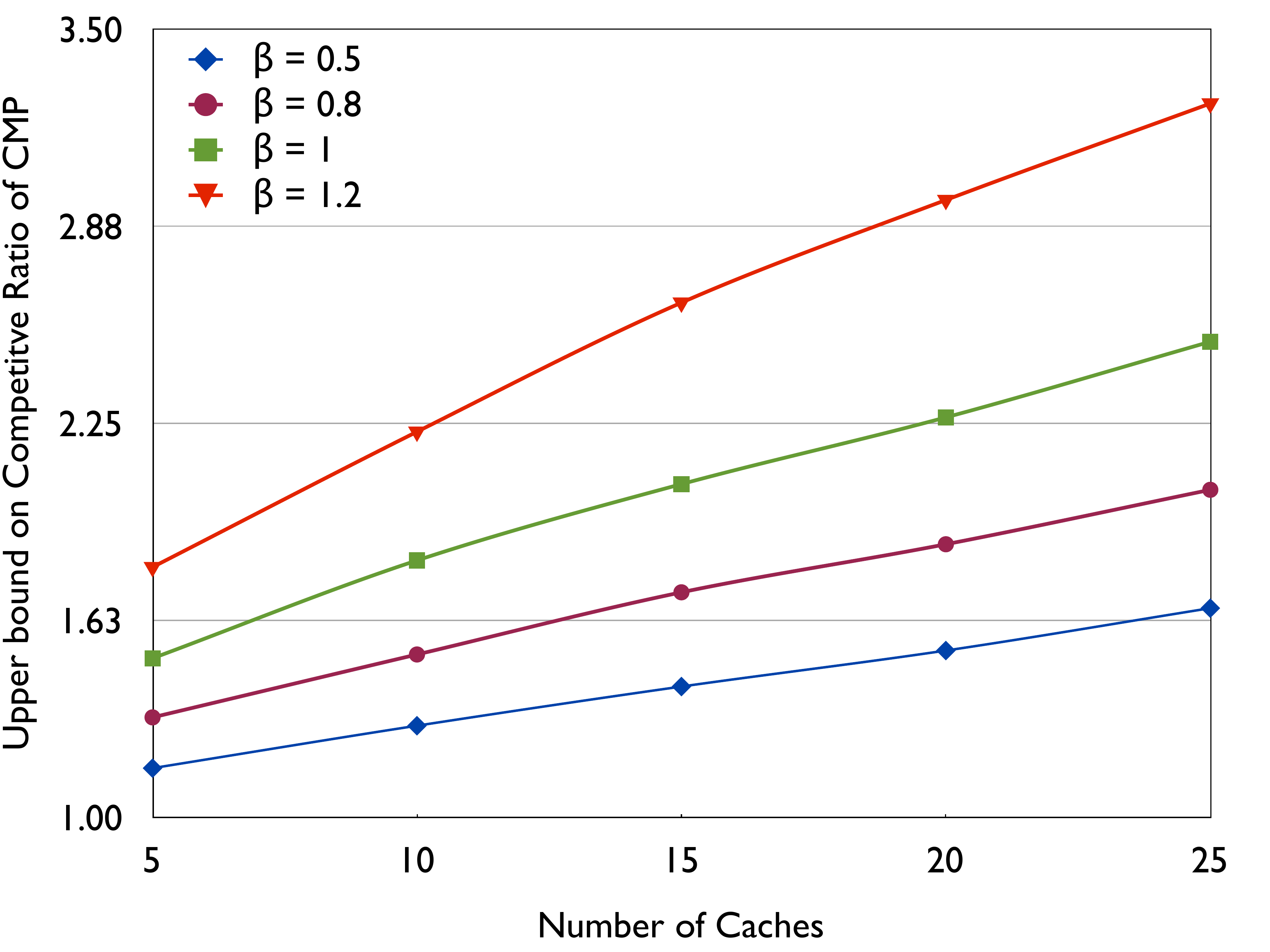}
		\caption{\sl The upper bound on the competitive ratio of the CMP policy for different number of caches $(m)$ and different values of the Zipf parameter $(\beta)$. \label{fig:diff_m_zipf}}  
	\end{center}
\end{figure}

In Figure \ref{fig:LRU_vs_CMP_zipf}, we compare the performance of the LRU policy with the CMP policy, when the underlying distribution is Zipf which is unknown to the LRU policy, while revealed for the CMP policy. We plot the fraction of requests which lead to faults as a function of the number of contents $(n)$ and the number of contents that can be stored in each cache $(k)$, with $m=10$.  The gap between the performance of the CMP policy and the LRU policy decreases as the Zipf parameter $(\beta)$ increases. For lack of space, we are unable to provide results for the time-correlated case, but they also follow the derived analytical results.
\begin{figure}[h]
	\begin{center}
		\includegraphics[scale=0.23]{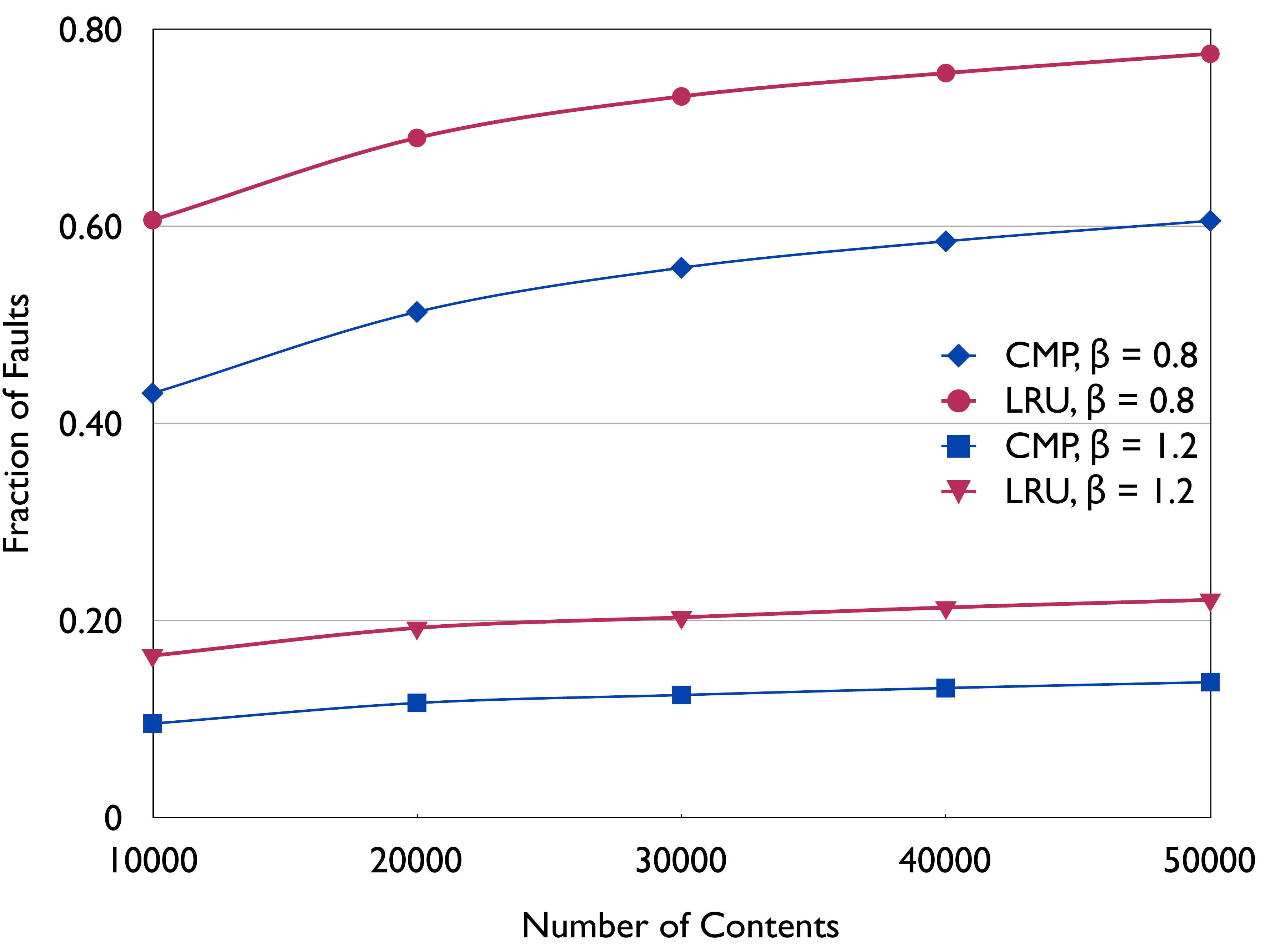}
		\caption{\sl The performance of the LRU policy and the CMP policy for different number of contents $(n)$ and different values of the Zipf parameter $(\beta)$. \label{fig:LRU_vs_CMP_zipf}} 
	\end{center}
\end{figure}

\section{Proofs}
\label{section:proof}

\subsection{Proof of Theorem \ref{thm:zipf_iid}}

\begin{lemma} 
	\label{lemma:upperbound_CMP}
	Let $\sfF_{\text{CMP}}(t)$ be the number of faults made by the CMP policy in time-slot $t$. Under Assumption \ref{ass:IID_arrivalprocess}, we have that, $m \sum_{i=k+1}^{n} p_i \leq \expect{\sfF_{\text{CMP}}(t)} \leq m \sum_{i=k}^{n} p_i.$
\end{lemma}
\begin{proof}[Proof of Lemma \ref{lemma:upperbound_CMP}]
	Let $\sfc_j(t)$ be the set of contents stored in cache $j$ at time $t$. Since each cache can store at most $k$ contents, $|\sfc_j(t)| = k$. Therefore, 
	\begin{eqnarray*}
		\expect{\sfF_{\text{CMP}}(t)} \geq \sum_{j=1}^m \bigg( \min_{\sfc_j(t)} \sum_{i \notin \sfc_j(t)} p_i \bigg) \geq m \sum_{i=k+1}^{n} p_i .
	\end{eqnarray*}
	
	Under the CMP policy, the $k-1$ most popular contents are stored on all $m$ caches at all times. Therefore, there can be a fault under the CMP policy only when $C_i$ for $i \geq k$ are requested. Since requests are i.i.d. across time and caches, the upper bound on $\sfF_{\text{CMP}}(t)$ follows.
\end{proof}

\begin{lemma} 
	\label{lemma:lowerbound_OPT}
	Let $\sfF_{\text{OPT}}(t)$ be the number of faults made by the optimal offline policy (OPT) in time-slot $t$. Under Assumption \ref{ass:IID_arrivalprocess}, we have that, $\expect{\sfF_{\text{OPT}}(t)} \geq m \sum_{i=mk+1}^{n} p_i$.
\end{lemma}

\begin{proof} 
	Let $\sfC(t)$ be the set of contents stored in the $m$ caches in time-slot $t$. It follows that $|\sfC(t)| \leq mk$. Each request for contents not in $\sfC(t)$ leads to a fault. Therefore,
	\begin{eqnarray*}
		\expect{\sfF_{\text{OPT}}(t)} \geq m \sum_{i \notin \sfC(t)} p_i \geq m \sum_{i=mk+1}^{n} p_i.
	\end{eqnarray*}
\end{proof}

%

\begin{proof} [Proof of Theorem \ref{thm:zipf_iid}] From Lemma \ref{lemma:upperbound_CMP} and \ref{lemma:lowerbound_OPT}, we get that,
	\begin{eqnarray}\label{eq:iid_competitiveratio}
		\dfrac{\expect{\sfF_{\text{ALG}}(R)}}{\expect{\sfF_{\text{OPT}}(R)}} \leq \dfrac{\sum_{i=k}^{n} p_i}{ \sum_{i=mk+1}^{n} p_i}.
	\end{eqnarray}
	
	From  \eqref{eq:iid_competitiveratio}, we have that, $\frac{\expect{\sfF_{\text{CMP}}(R)}}{\expect{\sfF_{\text{OPT}}(R)}} \leq e/E$ where $e = \sum_{i=k}^{n} p_i$ and $E =  \sum_{i=k+m-1}^{n} p_i.$ The results follow by computing the values of $e$ and $E$ for specific values of $\beta$.
	
	%
	%
	%
\end{proof}

%
%
\subsection{Proof of Theorem \ref{theorem:unknown_zipf}}
\begin{lemma}
	\label{lemma:E_1}	
	Let $\beta>1$ and $E_1$ be the event that in $\dfrac{k^{\beta}}{\log \log k}$ consecutive requests to a cache, no more than $\oo(k)$ different types of contents are requested. 
	Then, under Zipf distribution, for $k$ large enough, $\mathbb{P}(E_1^c) = \OO \big(e^{-\frac{k}{(\log \log k)^{1/\beta}}}\big)$.
\end{lemma}
\begin{proof}
	Recall that $p_i = \frac{i^{-\beta}}{\zeta(\beta)}$ for $\zeta(\beta) = \sum_{i=1}^n i^{-\beta}$.
	\begin{eqnarray*}
		\zeta(\beta) = \sum_{i=1}^{ n} i^{-\beta} &\geq& \int_{1}^{n+1} i^{-\beta} di
		\geq \frac{0.9}{\beta-1}
	\end{eqnarray*}
	for $n$ large enough.
	Therefore, for all $i$,
	$		p_i \leq \frac{\beta-1}{0.9} i^{-\beta}$.
	The total mass of all content types $i = \ell, \dots, n$ is
	\begin{eqnarray*}
		\sum_{i=\ell}^{n} p_i \leq \sum_{i=\ell}^{n} \frac{\beta-1}{0.9} i^{-\beta}
		\leq \int_{\ell-1}^{n} \frac{\beta-1}{0.9} i^{-\beta} di
		\leq \frac{1}{0.9} \dfrac{1}{(\ell-1)^{\beta-1}}.
	\end{eqnarray*}
	
	Now, for $\ell = \frac{k}{(\log \log k)^{1/\beta}}$ + 1, we have that,
	$
	\sum_{i=\ell}^{\alpha n} p_i  \leq \frac{1}{0.9} \frac{(\log \log k)^{1-1/\beta}}{k^{\beta-1}}.
	$
	Therefore, the expected number of requests for content types $\ell, \ell+1, \ldots, n$ is less than $\frac{1}{0.9} \frac{k}{(\log \log k)^{1/\beta}}$. Using the Chernoff bound, the probability that there are more than $\frac{2}{0.9} \frac{k}{(\log \log k)^{1/\beta}}$ requests for content types $\ell, \ell+1, \dots , n$ in the interval of interest is $\OO \big(e^{-\frac{k}{(\log \log k)^{1/\beta}}}\big)$. Hence, the result follows.
\end{proof}

\begin{lemma}
	\label{lemma:E_2}	
	Let $\beta>1$ and $E_2$ be the event that in $\dfrac{k^{\beta}}{\log \log k}$ consecutive requests to a cache, each one of the $\dfrac{k}{(\log k)^{2/\beta}}$ most popular contents is requested at least once. Then, for $k$ large enough, $\prob{E_2^c} \leq ke^{-\frac{c(\log k)^2}{\log \log k}}$,	where $c$ is a constant.
\end{lemma}
\begin{proof}
	Let $Q_i$ be the event that content $i$ is not requested in $\dfrac{k^{\beta}}{\log \log k}$ consecutive requests.
	$\prob{Q_i^c} = (1-p_i)^{\frac{k^{\beta}}{\log \log k}} = (1-ci^{-\beta})^{\frac{k^{\beta}}{\log \log k}}$.
	For $i \leq \dfrac{k}{(\log k)^{2/\beta}}$, $
	\prob{Q_i^c} \leq \bigg(1-\dfrac{c(\log k)^2}{\log \log k}\bigg)^{\frac{k^{\beta}}{\log \log k}} \leq e^{-\frac{c(\log k)^2}{\log \log k}}$.
	By the union bound over the  $\dfrac{k}{(\log k)^{2/\beta}}$ most popular contents, we have that,
	\begin{eqnarray*}
		\prob{E_2^c} \leq \dfrac{k}{(\log k)^{2/\beta}} e^{-\frac{c(\log k)^2}{\log \log k}} \leq ke^{-\frac{c(\log k)^2}{\log \log k}}.
	\end{eqnarray*}
\end{proof}

\begin{lemma} 
	\label{lemma:upperbound_LRU}
	Let $\sfF_{\text{LRU}}(t)$ be the number of faults made by the LRU policy in time-slot $t$. With Zipf distribution, we have that for $k$ large enough,
	\begin{eqnarray*}
		\expect{\sfF_{\text{LRU}}(t)} &=&  \Omega \bigg(\dfrac{m}{k^{\beta-1}}\bigg), \text{ and} \\
		\expect{\sfF_{\text{LRU}}(t)} &=& \OO \bigg(\dfrac{m}{k^{\beta-1}}(\log k)^{2-2/\beta}\bigg).
	\end{eqnarray*}
\end{lemma}

\begin{proof} 
	We first prove the lower bound on $\expect{F_{\text{LRU}}(t)}$. Let $\sfc_j(t)$ be the set of contents stored in cache $j$ at time $t$ under the LRU policy. Since each cache can store at most $k$ contents, $|\sfc_j(t)| = k$. Therefore, 
	\begin{eqnarray*}
		\expect{\sfF_{\text{LRU}}(t)} \geq \sum_{j=1}^m \bigg( \min_{\sfc_j(t)} \sum_{i \notin \sfc_j(t)} p_i \bigg) \geq m \sum_{i=k+1}^{n} p_i,
	\end{eqnarray*}
	thus proving the result. 
	
	The proof for the upper bound on $\expect{\sfF_{\text{LRU}}(t)}$ will be first conditioned on the events $E_1$ and $E_2$. Conditioned on $E_1 \cap E_2$, the  $\dfrac{k}{(\log k)^{2/\beta}}$ most popular contents are among the $k$ recently requested contents at time $t$. Given this, under the LRU policy, the  $\dfrac{k}{(\log k)^{2/\beta}}$ most popular contents are in the cache at time $t$. Therefore, there can be a fault under the LRU policy only when $C_i$ for $i > \dfrac{k}{(\log k)^{2/\beta}}$ are requested. Since requests are i.i.d. across caches, 
	$
	\expect{\sfF_{\text{LRU}}(t)|E_1 \cap E_2} \leq m \bigg( \sum_{i}^{n} p_i \bigg)
	$.
	In addition, 
	$\expect{\sfF_{\text{LRU}}(t)|(E_1 \cap E_2)^c} \leq m$. Therefore, 
	\begin{eqnarray*}
		\expect{\sfF_{\text{LRU}}(t)} \leq 
		m \bigg( \bigg( \sum_{i}^{n} p_i \bigg) + \prob{E_1^c} + \prob{E_2 ^c} \bigg).
	\end{eqnarray*}
	Hence, the result follows from Lemmas \ref{lemma:E_1} and \ref{lemma:E_2}, respectively, where $\prob{E_1^c}$ and $\prob{E_2 ^c}$ have been bounded. 
\end{proof}

\begin{proof} [Proof of Theorem \ref{theorem:unknown_zipf}] From Lemma \ref{lemma:lowerbound_OPT} and Lemma \ref{lemma:upperbound_LRU}, we have that, 	
	\begin{eqnarray*}
		\frac{\expect{\sfF_{\text{LRU}}(R)}}{\expect{\sfF_{\text{OPT}}(R)}} &\leq& \frac{\dfrac{ \tilde{c}(\log k)^{2-2/\beta}}{k^{\beta-1}} }{\dfrac{1}{(mk+1)^{\beta-1}} - \dfrac{1}{(n+1)^{\beta-1}}} 
	\end{eqnarray*}
	If $mk \leq \alpha n$ for a constant $\alpha<1$, we have that 
	\begin{eqnarray*}
		\liminf_{n \rightarrow \infty} \frac{\expect{\sfF_{\text{LRU}}(R)}}{\expect{\sfF_{\text{OPT}}(R)}} = \OO(m^{\beta-1}(\log k)^{2-2/\beta} ).
	\end{eqnarray*}
\end{proof}

\begin{lemma} 
	\label{lemma:upperbound_CMP_markov}
	Let $G_{\text{CMP}}$ be the number of faults made by the CMP policy while serving a request sequence. Under Assumption \ref{ass:correlated_general}, we have that, $ \sum_{i=k+1}^{n} \tilde{p}_i \leq \expect{G_{\text{CMP}}} \leq  \sum_{i=k}^{n} \tilde{p}_i$.
\end{lemma}

\begin{proof}
	Let $\sfc_j(t)$ be the set of contents stored in cache $j$ at time $t$. Since each cache can store at most $k$ contents, $|\sfc_j(t)| = k$. Therefore, 
	$
	\expect{G_{\text{CMP}}} \geq  \sum_{i=k+1}^{n} \tilde{p}_i$ .
	
	Under the CMP policy, the $k-1$ most popular contents are stored on all $m$ caches at all times. Therefore, there can be a fault under the CMP policy only when $C_i$ for $i \geq k$ are requested. The upper bound on $G_{\text{CMP}}$ follows.
\end{proof}

\begin{lemma} 
	\label{lemma:lowerbound_OPT_markov}
	Let $G_{\text{OPT}}$ be the number of faults made by the optimal offline policy (OPT) while serving a request sequence. Under Assumption \ref{ass:correlated_general}, we have that, $
	\expect{G_{\text{OPT}}} \geq  \sum_{i=m(k+\expect{L})+1}^{n} \tilde{p}_i,
	$
	where $L$ is the length of any request sub-sequence.
\end{lemma}

\begin{proof} 
	Let $\sfC(t)$ be the set of contents stored in the $m$ caches in time-slot $t$. It follows that $|\sfC(t)| \leq mk$. In addition,
	$|\sfC(t+1) \setminus \sfC(t)| \leq m$.
	Therefore,
	$
	\big|\cup_{w=t}^{t+L-1}\sfC(w)\big| \leq mk + mL$.
	Assuming the request sequence starts in time-slot $t$, and since the length of a request sequence is $L$ time-slots, each request for contents not in $\cup_{w=t}^{t+L-1}\sfC(w)$ leads to a fault. It follows that,
	\begin{eqnarray}
		\label{eq:pre_Jensen}
		\expect{G_{\text{OPT}}} \geq \expect{\sum_{i \notin \cup_{w=t}^{t+L-1}\sfC(w)} \tilde{p}_i}.
	\end{eqnarray}
	Using (\ref{eq:pre_Jensen}) and Jensen's inequality, $
	\expect{G_{\text{OPT}}} \geq  \sum_{i=m(k+\expect{L})+1}^{n} \tilde{p}_i$.\end{proof}

We now compare the performance of the CMP policy with the optimal policy (OPT) at the end of the first $T$ time-slots.

\begin{proof} [Proof of Theorem \ref{theorem:markov_competitiveratio}]
	By definition, 
	\begin{eqnarray*}
		\dfrac{\expect{\sfF_{\text{CMP}}(T)}}{\expect{\sfF_{\text{OPT}}(T)}} \leq \frac{\expect{\displaystyle  \sum_{r=1}^{Z(T)+m} G_{\text{CMP}}}}{\expect{\displaystyle \sum_{r=1}^{Z(T)} G_{\text{OPT}}}}. 
	\end{eqnarray*} 	
	The result then follows from Lemmas \ref{lemma:upperbound_CMP_markov} and \ref{lemma:lowerbound_OPT_markov}. 
\end{proof}

\section{Conclusions} In this paper, we extended the single cache paging problem to multiple caches, where multiple  simultaneous requests can be served at the same. We showed that this is a non-trivial extension, and that no online algorithm can be competitive with multiple caches in the worst case in contrast to the bounded competitive ratio of simple algorithms with a single cache. Then we analyzed two simple matching plus caching policies and showed that they have close to optimal performance for widely accepted stochastic input models.
\bibliographystyle{unsrt}
\bibliography{myref2}

\end{document}

%% file: input.tex
\usepackage{amsfonts}
\usepackage{times}
\usepackage{latexsym}
\usepackage{amssymb}
\usepackage{amsmath}
\usepackage{cite}
\usepackage{verbatim}
\newenvironment{proof}{ \textbf{Proof:} }{ \hfill $\Box$}

\def\bb0{{\mathbb{0}}}

\def\bb{{\mathbf{b}}}

\def\br{{\mathbf{r}}}

\def\b0{{\mathbf{0}}}

\def\b1{{\mathbf{1}}}

\def\sfC{\mathsf{C}}

\def\sfF{\mathsf{F}}

\def\sfc{{\mathsf{c}}}

\def\sf0{{\mathsf{0}}}
